\documentclass{www-notes}
\title{Warped Product Space-times}
\author{Xinliang An\Affiliation{University of Toronto, Toronto, ON; \url{xinliang.an@utoronto.ca}} \and Willie~Wai~Yeung Wong\Affiliation{Michigan State University, East Lansing, MI; \url{wongwwy@member.ams.org}}}

\newcommand*\hawking\varpi
\renewcommand*\covD[1][]{\nabla\ifthenelse{\equal{#1}{}}{}{^{[#1]}}}
\newcommand*\Ric[1]{\mathrm{Ric}^{[#1]}}
\newcommand*\Scal[1]{\mathrm{S}^{[#1]}}
\newcommand*\Einst[1]{\mathrm{G}^{[#1]}}
\newcommand*\secF{\mathit{II}}
\newcommand*\kodcur{\zeta}

\begin{document}
\maketitle

\begin{wwwabstract}
Many classical results in relativity theory concerning spherically symmetric space-times have easy generalizations to warped product space-times, with a two-dimensional Lorentzian base and arbitrary dimensional Riemannian fibers. 
We first give a systematic presentation of the main geometric constructions, with emphasis on the Kodama vector field and the Hawking energy; the construction is signature independent. 
This leads to proofs of general Birkhoff-type theorems for warped product manifolds; our theorems in particular apply to situations where the warped product manifold is not necessarily Einstein, and thus can be applied to solutions with matter content in general relativity.
Next we specialize to the Lorentzian case and study the propagation of null expansions under the assumption of the dominant energy condition. We prove several non-existence results relating to the Yamabe class of the fibers, in the spirit of the black-hole topology theorem of Hawking-Galloway-Schoen.
Finally we discuss the effect of the warped product ansatz on matter models. In particular we construct several cosmological solutions to the Einstein-Euler equations whose spatial geometry is generally not isotropic. 
\end{wwwabstract}

\section{Introduction}

In this paper we report on our investigation of pseudo-Riemannian warped product manifolds with 2 dimensional base. 
We prove, in this context, a microcosm of results that either generalize those previously were shown in the setting of a spherically symmetric ansatz in general relativity, or that specialize (with much simpler proofs) results known in greater generality. 

Our two main observations are:
\begin{enumerate}
	\item Many \emph{Lorentzian geometric} results relating to spherically-symmetric space-times do not, in fact, make full use of the spherical symmetry ansatz. In particular, they can be reproduced in the warped product context with essentially arbitrary fibers. These include the rigidity portion of Birkhoff's theorem, the realization of the Hawking energy as a dual potential of a geometric vector field, as well as much of causal theory under the dominant energy condition. 
	\item When coupled to \emph{matter fields} in the general relativistic setting, the warped product ansatz often impose strong restrictions on the allowable solutions, reducing the theory to be essentially similar to the spherically symmetric case. This include our novel generalized Birkhoff theorem classifying warped product electrovacuum solutions (previous results only treat vacuum solutions), as well as our observation that in the case of scalar field and fluid matter, the warped product ansatz forces the fiber manifold to be Einstein (but not necessarily homogeneous). 
\end{enumerate}

One of our original goals in the present paper was to develop a framework suitable for future investigations of the dynamical properties of solutions to Einstein's field equations, under the simplifying warped-product ansatz. The spherical symmetry ansatz has led to remarkable discoveries \cite{Christ1995,Christ1999,Daferm2003, Daferm2005, Daferm2012, BurLef2014}, and similarly planar or toroidal symmetries \cite{LefTch2011,LefSmu2010}. 
There has been furthermore recent results concerning general \emph{surface symmetric} spacetimes \cite{DafRen2016}.
One of our conclusions, unfortunately for our original motivation, is that for some common matter models, the warped-product ansatz is \emph{no more general} than the surface symmetric ansatz. 

Below we first introduce the warped product ansatz and then summarize the main results of this paper. 

\subsection{Pseudo-Riemannian warped products}
Let $(Q,g)$ be a two-dimensional pseudo-Riemannian manifold and $(F,h)$ be a pseudo-Riemannian manifold with dimension $n \geq 2$. We consider the warped product $M = Q\times_r F$ where $r$ is a positive real-valued function on $Q$. The warped product metric on $M$ is given by 
\begin{equation}
	\tilde{g} = g + r^2 h,
\end{equation}
where by an abuse of notation we identify $g$ and $h$ with their pull-backs onto $M$ via the canonical projections. We refer to $r$ as either the \emph{warping function} or, following the relativity literature, the \emph{area radius}. 

Throughout we will use $\covD[\tilde{g}]$, $\covD[g]$, and $\covD[h]$ to denote the Levi-Civita connection on $(M,\tilde{g})$, $(Q,g)$, and $(F,h)$ respectively; when there is no risk of confusion (e.g.\ for scalars) we will just use $\covD$. Similarly $\Ric{\tilde{g}}$, $\Scal{g}$ etc.\ will denote the corresponding Ricci and scalar curvatures, and $\Einst{h}$ etc.\ will denote the corresponding Einstein tensors $\Einst{h} = \Ric{h} - \frac12 \Scal{h} h$. 

In this subsection we record some basic, well-known properties of warped-product manifolds; see \cite[Ch.\ 7, \P 33--43]{ONeil1983} and \cite[Ch.\ 3, \S6]{BeEhEa1996} for detailed proofs and computations (note that in the latter the authors define the warping function as what would be $r^2$ in our notation).

\begin{prop}[Connection properties {\cite[\P7.35]{ONeil1983}}] 
	On $M = Q \times_r F$, if $X, Y$ are tangent to $Q$ and $V,W$ are tangent to $F$, we have:
	\begin{enumerate}
		\item The second fundamental form of $F$ as embedded submanifolds in $M$ satisfy
			\[
				\secF(V,W) = - h(V,W)\cdot r (\D*r)^\sharp = - \tilde{g}(V,W) \cdot \frac{(\D*r)^\sharp}{r}.
			\]
			In particular the fibers $\{q\}\times F$ are totally umbilic with mean curvature $-\frac{n}{r} (\D*r)^\sharp$. 
		\item $\covD[\tilde{g}]_X Y= \covD[g]_XY$ (where we identified the latter with its lift). 
		\item $\covD[\tilde{g}]_X V = \covD[\tilde{g}]_V X = \dfrac{X(r)}{r} V$.
		\item $\covD[\tilde{g}]_V W + \innerprod[h]{V}{W} r (\D*r)^\sharp = \covD[h]_VW$. 
	\end{enumerate}
\end{prop}

\begin{prop}[Curvature properties {\cite[\P7.44]{ONeil1983}}]
	On $M = Q\times_r F$, given $X,Y$ tangent to $Q$ and $V,W$ tangent to $F$, we have:
	\begin{subequations}
	\begin{align}
		\Ric{\tilde{g}}(X,Y) &= \Ric{g}(X,Y) - \frac{n}{r} \covD^2_{X,Y} r; \label{eq:RicciQ}\\
		\Ric{\tilde{g}}(X,V) &= 0; \label{eq:RicciCross} \\
		\Ric{\tilde{g}}(V,W) &= \Ric{h}(V,W) - \left( r\triangle_g r + (n-1) \innerprod[g]{\D*{r}}{\D*{r}} \right) h(V,W).\label{eq:RicciF}
	\end{align}
	\end{subequations}
\end{prop}

\begin{conv}
	We let $\covD^2 r$ denote the lift to $M$ of the Hessian of $r$ on $(Q,g)$. The operator $\triangle_g = \trace_g \covD^2$ is the Laplace-Beltrami operator of the metric $g$; depending on its signature the operator could be either elliptic or hyperbolic (recall that $Q$ is two-dimensional). 
\end{conv}

An immediate consequence of the curvature properties is that the Einstein tensor is block-diagonal. More precisely, we can write $\Einst{\tilde{g}} = G_Q + G_F$ where $G_Q$ is tangent to $Q$ and $G_F$ is tangent to $F$, and are defined by
\begin{subequations}
	\begin{align}
		G_Q &\eqdef - \frac{n}{r} \covD^2 r + \frac{n}{r} \triangle_g r g \label{eq:EinsteinQ}\\
		\notag & \qquad\qquad - \frac{1}{2r^2} \left[ \Scal{h} - n(n-1) \innerprod[g]{\D*{r}}{\D*{r}} \right] g;\\
		G_F &\eqdef \Einst{h} - \frac{r^2}2 \Scal{g} h \label{eq:EinsteinF}\\
		\notag & \qquad\qquad + (n-1)r\left[ \triangle_g r + \frac{n-2}{2} \innerprod[g]{\D*{r}}{\D*{r}} \right]h.
	\end{align}
\end{subequations}
Here we used that since $(Q,g)$ is two-dimensional, its corresponding $\Einst{g}\equiv 0$. Note that in the case $n = 2$ (such as in the case of spherically symmetric solutions in $(1+3)$-dimensional gravity), the last equation simplifies also with $\Einst{h} \equiv 0$.

\subsection{Summary of results}
In Section \ref{sec:KodBirHaw} we introduce definitions, in the setting of warped-product manifolds, of the \textbf{Kodama vector field} and \textbf{Hawking energy density}. These generalize their well-known counterparts in spherical symmetry. We remark here that we make no a priori assumptions on the fiber $(F,h)$ being symmetric. The main results proven in this section are \textbf{Propositions \ref{prop:birkhoff1}} and \textbf{\ref{prop:hawkpot}}. 
The former is a generalization of Birkhoff's rigidity theorem; it states roughly that if $G_Q$ is pure-trace then the base manifold $(Q,g)$ admits a Killing vector field. Our hypothesis is weaker than previous results in this direction \cite{Schmid1997}. 
The latter states that the Hawking energy is the dual potential for the energy current associated to the Kodama vector field, a result only previously shown in spherical symmetry \cite{Kodama1980}. 

In Section \ref{sec:Fisotropy} we consider the consequences that $G_F$ is pure-trace. The main result is \textbf{Theorem \ref{thm:birkhoff2}} which generalizes the classification portion of Birkhoff's theorem. Roughly speaking, we are able to fully classify manifolds for which $G_Q$ and $G_F$ are both pure trace (note that our assumptions do not require $(M,\tilde{g})$ to be Einstein). The allowable manifolds are special Cartesian products, certain pp-wave space-times, and generalizations of the Reissner-Nordstr\"om(-(A)dS) family. Note that in this and in the previous section we allow our manifold to have arbitrary signature. 

We restrict to the Lorentzian case in Section \ref{sec:Causal}, and study the causal propagation properties in conjunction with the dominant energy condition. In \textbf{Theorem \ref{thm:topology}} we give a very short proof of the black hole topology theorem in our limited setting, and complement it with several other non-existence results that depends on the topology of the fiber $F$ (see \textbf{Theorem \ref{thm:negcurperiod}} and \textbf{Proposition \ref{prop:negcuropen}}). 

Finally in Section \ref{sec:Matter} we consider the effects of the warped product assumption on matter models. The main thrust of the results is that matter models force additional rigidity into the problem, and lead us to results similar to the spherically symmetric cases.

\section{Kodama, Birkhoff, and Hawking}\label{sec:KodBirHaw}
On the manifold $(Q,g)$, making use of its two-dimensionality we can write down the \emph{Kodama vector field}
\begin{equation}\label{eq:Kdef}
	K\eqdef (*_g \D*{r})^{\sharp}
\end{equation}
where $*_g$ is the Hodge operator (see Appendix \ref{app:hodgeprop} for a summary of its properties and some related notations). $K$ lifts to a vector field on $M$ that is orthogonal to the leaves $F$; we abuse notation and call the lift also $K$. 
This vector field was first constructed for spherically symmetric solutions in general relativity by Kodama \cite{Kodama1980} who realized that the Hawking energy can be described as an energy flux for this vector field. 
Recently it has been also used as a replacement for the time-like Killing vector field when trying to analyze \emph{non-static} spherically symmetric solutions to the Einstein equations (see \cite{AbrVis2010} and references therein).  
The goal of this section is to obtain generalizations of the above two uses to the context of pseudo-Riemannian warped products. 

Let $X$ be a vector field on $(Q,g)$, identified with its lift to $M$. Its deformation tensor on $(M,\tilde{g})$ is by definition
\[
	\lieD_X \tilde{g} = \lieD_X g + \lieD_X (r^2 h).
\]
By the warped product construction, $\lieD_X g$ is the same as the lift of the corresponding object on $Q$, while $\lieD_X h = 0$. So we obtain that
\begin{equation}\label{eq:defdecom}
	\lieD_X \tilde{g} = \lieD_X g + 2 r X(r) h.
\end{equation}
An immediate consequence is that (see also \cite[Remark 3.57.(7)]{BeEhEa1996})
\begin{lem}\label{lem:Kdivfree}
	A Killing vector field $X$ on $(Q,g)$ lifts to a Killing vector field on $(M,\tilde{g})$ if and only if $X(r) = 0$. 
\end{lem}

The Kodama vector field is not Killing on $(Q,g)$ in general; however, $K$ is divergence free (since $\mathrm{d}^2 = 0$). This implies that $\trace_g \lieD_K g = 0$. 
Furthermore, by its definition $K(r) = 0$, and so we conclude the following generalization of Kodama's result in spherical symmetry \cite{Kodama1980}.
\begin{lem}
	The deformation tensor of $K$ as a vector field on $(M,\tilde{g})$ is the lift of its deformation tensor as a vector field on $(Q,g)$, that is to say, $\lieD_K \tilde{g} = \lieD_K g$. Furthermore, $\trace_{\tilde{g}} \lieD_K \tilde{g} = 0$. 
\end{lem}

Let us examine $\lieD_K \tilde{g}$ in more detail. On the base manifold we have
\begin{equation}\label{eq:defKalt}
	(\lieD_K g)(X,Y) = \innerprod[g]{\covD_X K}{Y} + \innerprod[g]{\covD_Y K}{X} = - (\covD^2 r)(X, *_g Y) - (\covD^2 r)(*_g X, Y)
\end{equation}
by \eqref{eq:Kdef}. Therefore by the first part of Proposition \ref{prop:hodge2} we have
\[
	(\lieD_K g)(X, *_g Y) =  -\sign(g) (\covD^2 r)(X, Y) - (\covD^2 r)(*_g X, *_g Y).
\]
Since the Hessian is symmetric, we can apply the second part of Proposition \ref{prop:hodge2} to get
\begin{equation}\label{eq:deformK}
	(\lieD_K g)(X, *_g Y) = -\sign(g) \left[ 2 (\covD^2 r)(X,Y) - \triangle_g r \innerprod[g]{X}{Y}\right].
\end{equation}
Note that inside the bracket on the right hand side is twice the trace-free part of the Hessian of $r$. This implies, in particular, the following version of Birkhoff's rigidity theorem. 
\begin{prop}[Birkhoff's theorem for warped products: rigidity] \label{prop:birkhoff1}
	Let $M = Q\times_r F$ be a pseudo-Riemannian warped product with $Q$ two-dimensional and connected. Suppose the Einstein tensor is such that $G_Q$ is pure-trace relative to $g$ (see \eqref{eq:EinsteinQ} and the preceding paragraph), then either the Kodama vector field $K$ is a non-trivial Killing vector field on $M$, or $r$ is constant (and hence $M$ is a Cartesian product).
\end{prop}
\begin{proof}
	By \eqref{eq:EinsteinQ} we see that the trace-free part of $G_Q$ relative to $g$ is precisely the trace-free part of the Hessian of $r$, and hence by \eqref{eq:deformK} the Kodama vector field is Killing. As is well-known (see \cite[\P9.27]{ONeil1983}), if $K(p)$ and $\covD K(p)$ both vanish at some $p\in M$, than $K\equiv 0$ on $M$. Therefore either $K$ is non-vanishing on an open dense subset of $M$, or that $K\equiv 0$ and hence $\D*{r} \equiv 0$.
\end{proof}

\begin{rmk}
	The original Birkhoff theorem states that the only spherically symmetric solutions of the vacuum (electrovac) Einstein equations are the Schwarzschild (Reissner-Nordstr\"om) family of solutions, which in addition to being spherically symmetric, are also static. Modern interpretations of Birkhoff's theorem often restate this in terms of the presence of an additional Killing vector field transverse to the fiber $F$, as we do here. 
	Similar rigidity statements for warped-product manifolds have previously been proven under the assumption that $M$ is \emph{Ricci-flat} (equivalently $\Einst{\tilde{g}} \equiv 0$) in \cite{Schmid1997, DobUna2008}. Our hypotheses are much weaker and are suitable for applications in relativity (see, e.g.\ Corollary \ref{cor:BirkhoffEV}). 
\end{rmk}

\begin{rmk}\label{rmk:discreteFix}
	We mention in passing that, under the assumptions of Proposition \ref{prop:birkhoff1}, when $\D*{r} \not\equiv 0$, then the subset $\{\D*{r} = 0\} \subset Q$ must be discrete. Killing's equation implies that $\covD (K^\flat)$ is purely anti-symmetric; since $Q$ is two-dimensional, this quantity must be proportional to the volume form. Supposing $\{\D*{r} = 0\}$ has a limit point $p$, then there would exist a vector $V\in T_pM \setminus \{0\}$ such that $\covD_V K|_p = 0$. This implies the constant of proportionality is $0$ at $p$, and hence $\covD K|_p = 0$. This then implies $K$ vanishes identically.
\end{rmk}

We conclude this section with a discussion of a generalization of the Hawking (or Misner-Sharp) energy to warped product space-times. Since we do not assume that $(F,h)$ is symmetric, it is most convenient to work in terms of the corresponding \emph{energy density}. 
\begin{defn}[Hawking energy]
	Let $M = Q\times_r F$ be a pseudo-Riemannian warped product. Consider the function $\hawking$ given by the expression
	\begin{equation}\label{eq:varpidef}
		\hawking \eqdef  \frac{r^{n-1}}{2} \left[ \frac{\Scal{h}}{n(n-1)} - \innerprod[g]{\D*{r}}{\D*{r}} \right]. 
	\end{equation}
	The $n$-form $\hawking\cdot\mathrm{dvol}_h$ is the \emph{Hawking energy density}, where $\mathrm{dvol}_h$ is the pull-back to $M$ of the volume form of the metric $h$ on $F$. 
\end{defn}
\begin{rmk}
	The normalizing factor $n(n-1)$ is the scalar curvature of the standard unit sphere $\Sphere^n$. Note also that a factor of $\hawking$ appears in the second line of \eqref{eq:EinsteinQ}. 
\end{rmk}
In the case of spherical symmetry, Kodama observed that the Hawking energy is the dual potential to the conserved current generated by his namesake vector field \cite{Kodama1980}. Here we give a generalization to pseudo-Riemannian warped products. 
\begin{defn}[Kodama current]
 	The \emph{Kodama current} $\kodcur$ is the one-form defined on $M$ by the expression
	\[ \kodcur(X) = \Einst{\tilde{g}}(X, K).\]
\end{defn}
\begin{rmk}
	We note that $\kodcur$ is tangential to $Q$, since by the warped product assumption $\Einst{\tilde{g}}$ is block diagonal.
\end{rmk}
\begin{prop}
	The Kodama current is divergence free. 
\end{prop}
\begin{proof}
	We have that $\mathrm{div}_{\tilde{g}} \kodcur = (\mathrm{div}_{\tilde{g}} \Einst{\tilde{g}})(K) + \innerprod[\tilde{g}]{\Einst{\tilde{g}}}{\covD K}$. 
	By the Bianchi identities, the Einstein tensor is divergence free, and the first term vanishes. Since the Einstein tensor is symmetric, the second term can be re-written as $\innerprod[\tilde{g}]{\Einst{\tilde{g}}}{\frac12 \lieD_K \tilde{g}} = \innerprod[g]{G_Q}{\frac12 \lieD_K g}$ by \eqref{eq:defdecom}. (Recall that $G_Q$ is the block of $\Einst{\tilde{g}}$ tangent to $Q$.)
	
	Now, by Lemma \ref{lem:Kdivfree} we have $\lieD_K g$ is $g$-trace-free. As we have already seen, by \eqref{eq:EinsteinQ} the $g$-trace-free part of $G_Q$ is precisely given by the trace-free part of the Hessian $\covD^2 r$, so we get
	\[ \mathrm{div}_{\tilde{g}} \kodcur = - \frac{n}{2r} \innerprod[g]{\covD^2 r}{\lieD_K g} = - \frac{n}{2r} \innerprod[g]{\covD^2 r}{*_g (\covD^2 r) + (\covD^2 r)*g} = 0,\]
	where the second equality is due to \eqref{eq:defKalt} (see Appendix \ref{app:hodgeprop} for notation). 
\end{proof}
\begin{prop}\label{prop:hawkpot}
	The Hawking energy density is the dual potential for the Kodama current:
	\[ \D*(\hawking~\mathrm{dvol}_h) = - \sign(g) \frac{r^n}n (*_g \kodcur)\wedge \mathrm{dvol}_h .\]
\end{prop}
\begin{proof}
	Since $\mathrm{dvol}_h$ is a top-form on $F$, the proposition follows if we can show that for every $X$ tangent to $Q$, 
	\[ X(\hawking) = \sign(g)  \frac{r^n}{n} \kodcur(*_g X). \]
	We compute
	\begin{align*}
		X(\hawking) &= \frac{n-1}{r} \hawking X(r) -  r^{n-1} \nabla^2 r( \D*{r}^\sharp, X), \\
			&= r^{n-1} \left[ (n-1) \frac{\hawking}{r^n} \innerprod[g]{\D*{r}^\sharp}{X} - \nabla^2 r(\D*{r}^\sharp,X) \right].\\
		\intertext{Noting that the Hessian $\nabla^2 r$ is symmetric, we can apply Proposition \ref{prop:hodge2} to get}
			&=  - \sign(g) r^{n-1} \left[ (n-1) \frac{\hawking}{r^n} \innerprod[g]{K}{*_g X} + \nabla^2r (K, *_g X) - \triangle_g r \innerprod[g]{K}{*_g X}\right].
	\end{align*}
	The term inside the square brackets is precisely $-\frac{r}{n} G^{[\tilde{g}]}(K, *_g X)$ as in \eqref{eq:EinsteinQ}. 
\end{proof}
\begin{rmk}
	Note that we can equivalently write $r^n (*_g \kodcur) \wedge \mathrm{dvol}_h$ as $*_{\tilde{g}} \kodcur$, the space-time Hodge dual of the Kodama current. 
\end{rmk}

\section{Pointwise isotropy of $G_F$ on $F$}\label{sec:Fisotropy}
In the previous section we have already explored the consequence if $G_Q$ is assumed to be proportional to $g$ (Proposition \ref{prop:birkhoff1}); note that the coefficient of proportionality is, in general, allowed to be a function on $M$ that can depend on the position on $F$. In this section we will collect some consequences of the assumption that $G_F$ instead is proportional to $h$.

\begin{lem}[Rigidity] \label{lem:rigid}
	Let $M = Q\times_r F$ be s pseudo-Riemannian warped product, and suppose that $G_F = \lambda h$ for some $\lambda: M\to\Real$. Then $\lambda$ is in fact independent of $F$. Furthermore, if $n \geq 3$, then $(F,h)$ is necessarily Einstein.
\end{lem}
\begin{proof}
	Combining the assumption on $G_F$ with \eqref{eq:EinsteinF}, we obtain
	\[ G^{[h]} - \frac{r^2}{2} S^{[g]} h + (n-1) r \left[ \triangle_g r + \frac{n-2}{2} \innerprod[g]{\D*{r}}{\D*{r}} \right] h = \lambda h.\]
	In the case $n = 2$, the Einstein tensor $G^{[h]} \equiv 0$, and as $r$ and $S^{[g]}$ are all constant along $F$, so must $\lambda$. 
	
	In the case $n \geq 3$, noting that the above expression implies that $\Ricci^{[h]}$ is point-wise proportional to the metric $h$, we have that $(F,h)$ is necessarily Einstein (see \cite[Ch.\ 3, exer.\ 21]{ONeil1983}) and hence both $S^{[h]}$ and $\lambda$ are constant along $F$. 
\end{proof}

Assuming that $G_Q$ and $G_F$ are both pure trace, we can fully classify pseudo-Riemannian warped products under some assumptions on their coefficients of proportionality. This is one of the main results of this paper and is a generalization of the classification part of Birkhoff's theorem. 
\begin{thm}[Birkhoff's theorem for warped products: classification] \label{thm:birkhoff2}
	Let $M = Q\times_r F$ be a pseudo-Riemannian warped product, where $Q$ has dimension 2 and is connected. Suppose there exist constants $\Lambda$ and $\kappa$, as well as a function $p$ such that 
	\[ G_Q = (\kappa p - \Lambda)g, \qquad G_F = (p - \Lambda) r^2 h. \]
	Then, defining
	\begin{subequations}
	\begin{equation}\label{eq:Pidef}
		\Pi(r) = \begin{cases} \frac{\kappa}{n+\kappa} r^{(n+\kappa)/\kappa} & \kappa \neq 0, -n \\
		\ln(r) & \kappa = -n \\
		0 & \kappa = 0\end{cases},
	\end{equation}
	we have that there exist constants $p_0$ and $m$ such that 
	\begin{equation}\label{eq:pprop}
		p^\kappa = p_0^\kappa r^{n(1-\kappa)},
	\end{equation}
	and
	\begin{equation}\label{eq:Knorm}
		\innerprod[g]{\D*{r}}{\D*{r}} = -\sign(g) \innerprod[g]{K}{K} = \frac{\Scal{h}}{n(n-1)} - \frac{2m}{r^{n-1}} - \frac{2\Lambda}{n(n+1)} r^2 + \frac{2\kappa p_0}{n} \Pi\cdot  r^{1-n}.
	\end{equation}
	In addition 
		\begin{equation}\label{eq:Kvf}
			\text{there exists a non-trivial, $Q$-tangential Killing vector field on $(M,\tilde{g})$}.
		\end{equation}
	\end{subequations}
	Finally, exactly one of the following holds:
	\begin{enumerate}
		\item The warping function $r$ is constant. $M$ is a Cartesian product $Q\times F$ here $F$ is Einstein and $Q$ is maximally symmetric. (In other words, $Q$ is either locally isometric to the Euclidean plane $\Real^2$, the Minkowski plane $\Real^{1,1}$, or one of the four two-dimensional hyperquadrics.)
		\item The warping function $r$ is not constant, but $\innerprod[g]{\D*{r}}{\D*{r}} \equiv 0$. $(Q,g)$ is locally isometric to the Minkowski plane $\Real^{1,1}$ with $\covD^2 r \equiv 0$, and we have either
			\begin{enumerate}
				\item $(M, \tilde{g})$ and $(F,h)$ are both Ricci-flat; or
				\item $n \geq 3$, $\kappa = \frac{n}{n-2}$, and $m = \Lambda = 0 = \Scal{h} + 2 \kappa p_0$. 
			\end{enumerate}
		\item The set $\{\innerprod[g]{\D*{r}}{\D*{r}} \neq 0\}$ is open and dense in $Q$; on this set the metric $g$ can be expressed in local coordinates as
			\begin{subequations}
			\begin{equation}\label{eq:metricformorth}
				\frac{1}{\innerprod[g]{\D*{r}}{\D*{r}}} ~\D r^2 + \innerprod[g]{K}{K} ~\D s^2,
			\end{equation}
			where $r$ is the warping function and the metric coefficients are given explicitly by \eqref{eq:Knorm}.
			Furthermore, in the case where $g$ is Lorentzian, on the set $\{\D*{r} \neq 0\}$ the metric $g$ can be expressed in local coordinates as
			\begin{equation}\label{eq:metricformnull}
				- \innerprod[g]{\D*{r}}{\D*{r}} \D{u}^2 \pm \left( \D*{u} \otimes \D*{r} + \D*{r} \otimes \D*{u}\right).
			\end{equation}
			\end{subequations}
	\end{enumerate}
\end{thm}
\begin{proof}
	First, a few immediate consequences:
	\begin{itemize}
		\item From Lemma \ref{lem:rigid} we have that $p$ is constant along $F$. This implies also via \eqref{eq:EinsteinQ} that $\Scal{h}$ is constant, and hence for all $n \geq 2$ we can conclude that $(F,h)$ is Einstein. 
		\item By Proposition \ref{prop:birkhoff1} the Kodama vector field is Killing. So either $\D*{r} = 0$ everywhere, or $\{\D*{r} = 0\}$ is discrete (see Remark \ref{rmk:discreteFix}).
	\end{itemize}
	Under our hypotheses, the Einstein tensor takes the form 
	\[ \Einst{\tilde{g}} = (\kappa - 1)p g + (p-\Lambda) \tilde{g}.\]
	By Bianchi's identity $\Einst{\tilde{g}}$ is divergence free. Since
	\[ \mathrm{div}_{\tilde{g}} g = - (\trace \secF)^\flat = \frac{n}{r} \D*{r},\]
	we have that
	\begin{equation}\label{eq:interEq1}
		0 = \mathrm{div}_{\tilde{g}} \Einst{\tilde{g}} = (\kappa - 1)p \frac{n}{r} \D{r} + \kappa \D{p} .
	\end{equation}
	
	First consider the case $\D*{r}\equiv 0$. By \eqref{eq:interEq1} we have that $p$ must be constant and hence \eqref{eq:pprop} holds. 
	With an appropriate choice of the constant $m$ then \eqref{eq:Knorm} also holds trivially, being an algebraic relation. 
	Finally by \eqref{eq:EinsteinF} the base manifold $(Q,g)$ has constant scalar curvature, and hence is Einstein and maximally symmetric. This implies both the first of the trichotomy options, as well as \eqref{eq:Kvf}. 

	Now suppose $\D*{r}$ does not vanish identically. Then $K$ is non-trivial and \eqref{eq:Kvf} holds. If $\kappa = 0$, then \eqref{eq:interEq1} implies that $p \equiv 0$, in which case $(M,\tilde{g})$ is Einstein with $\Einst{\tilde{g}} = - \Lambda \tilde{g}$. When $\kappa \neq 0$ and $p\not\equiv 0$, we have that \eqref{eq:interEq1} implies 
	\[ \frac{1-\kappa}{\kappa} n \D{(\ln r)} = \D{(\ln p)} \implies p = p_0 r^{n(1-\kappa)/\kappa}\]
	for some constant $p_0$. This verifies \eqref{eq:pprop}. It also further implies $G_Q = (\kappa p_0 r^{n(1-\kappa)/\kappa} - \Lambda) g$ from which we conclude
	\[ \kodcur = (\kappa p_0 r^{n(1-\kappa)/\kappa} - \Lambda) (*_g \D*{r}). \]
	Therefore
	\begin{equation}
		\sign(g) *_{g} \kodcur = \left[ \Lambda r^n - \kappa p_0 r^{n/\kappa} \right] \D{r} 
		= \D* \left(\frac{\Lambda}{n+1} r^{n+1} - \kappa p_0 \Pi \right), 
	\end{equation}
	where $\Pi$ is as defined in \eqref{eq:Pidef}. By Proposition \ref{prop:hawkpot} we have then 
	\begin{equation}
		m \eqdef \hawking - \frac{\Lambda}{n(n+1)} r^{n+1} + \frac{\kappa p_0}{n} \Pi
	\end{equation}
	defines a constant. Going back to \eqref{eq:varpidef} and rearranging we see that \eqref{eq:Knorm} also hold (using in the course Proposition \ref{prop:hodge2}).

	Now, since we assume $\D*{r}$ does not vanish identically, $K$ is non-zero on an open and dense subset of $Q$. We claim that the set $\{\innerprod[g]{K}{K} = 0\}$ either has empty interior or is equal to the whole of $Q$. By \eqref{eq:Knorm}, on this set we must have
	\begin{equation}\label{eq:intereq2}
		\frac{\Scal{h}}{n(n-1)} - \frac{2m}{r^{n-1}} - \frac{2\Lambda}{n(n+1)} r^2 + \frac{2\kappa p_0}{n} \Pi \cdot r^{1-n} = 0
	\end{equation}
	Since the functions $1, r^{1-n},$ and $r^2$ are linearly independent over an interval, and $\Pi \cdot r^{1-n} = 0$ can at most be in the linear span of one of them, this implies either $r$ is constant on $\{\innerprod[g]{K}{K} = 0\}$, or \eqref{eq:intereq2} is an identity that holds for all $r$. Our claim follows. 

	Suppose now we are in the case where $K$ is null on the whole of $Q$ but $\D*{r}$ does not vanish identically. This implies 
	\begin{enumerate}
		\item $g$ is Lorentzian,
		\item $\covD \innerprod[g]{K}{K} \equiv 0$, and 
		\item \eqref{eq:intereq2} must hold for all $r$. 
	\end{enumerate}
	Using that $\covD (K^\flat)$ is antisymmetric, the second item above gives that $\covD_K K = 0$; the same argument as Remark \ref{rmk:discreteFix} then implies $\covD K \equiv 0$, or $\covD^2 r \equiv 0$. Furthermore, by the Jacobi equation for Killing vector fields \cite[Ch.\ 9, exer.\ 8]{ONeil1983} this implies $\Ric{g}(K,\text{---}) = 0$. As $Q$ is two dimensional this implies $Q$ is flat. 
	Next, examining \eqref{eq:Pidef}, we see that for \eqref{eq:intereq2} to hold for all $r$, necessarily $m \equiv 0$.
	Furthermore, for generic $\kappa$, $\Einst{\tilde{g}}$ must vanish identically as $\kappa p - \Lambda$ and $p - \Lambda$ must both vanish; the only exception being the case when $n \geq 3$ and $\kappa = \frac{n}{n-2}$, in which case $\Lambda = 0$ and $2 \kappa p_0 + \Scal{h} = 0$. In the generic case this implies $(M,\tilde{g})$ is Ricci-flat, and by \eqref{eq:RicciF} (together with the fact that $\covD^2 r \equiv 0$) so is $(F,h)$. This proves the second of the trichotomy options. 
	
	It remains to consider the case where $\{ K \text{ is null} \}$ has empty interior and derive expressions for the metric. 
	In a neighborhood where $K$ is not null, the vector fields $K$ and $\covD r$ are independent and orthogonal. Let $s$ be a local solution to the equation $\innerprod[g]{\D*{r}}{\D*{s}} = 0$, i.e.\ $s$ is constant along the integral curves of $\covD r$. 
	Since $K$ is Killing and $K(r) = 0$, we have
	\[ 0 = \lieD_K \innerprod[g]{\D*{r}}{\D*{s}} = \covD_{\covD r} \lieD_K s,\]
	implying $K(s)$ is constant along the level sets of $s$. Therefore by reparametrizing we can assume without loss of generality that $s$ satisfies $K(s) \equiv 1$, at least locally, or that $K = \partial_s$ in the $(r,s)$ coordinate system. Therefore in this coordinate system the metric $g$ has the line element exactly as given in \eqref{eq:metricformorth}. 

	When $(Q,g)$ is Lorentzian, the conformal structure uniquely defines a double-null foliation of $Q$. On sufficiently small domains on which $K$ does not vanish, we can select a function $u$ such that $\innerprod[g]{\D*{u}}{\D*{u}} \equiv 0$ and $\innerprod[g]{\D*{u}}{\D*{r}}$ never vanishes. Let $L \eqdef (\D*{u})^\sharp$, we have
	\[ 0 = \lieD_K \innerprod[g]{\D*u}{\D*u} = 2 L(K(u)),\]
	implying that $K(u)$ is a function of $u$ only. And so by reparametrizing $u$ we can assume that $K(u) \equiv 1$ on our small domain. Next, observe that since $\innerprod[g]{K}{K} = - \innerprod[g]{\D*{r}}{\D*{r}}$ as $g$ is by assumption Lorentzian, and by construction $K(r) = 0$, we have that the two vectors $K \pm (\D*{r})^\sharp$ are null, and hence one of them is parallel to $L$. This implies that 
	\[ \pm \innerprod[g]{\D*{u}}{\D*{r}} = \innerprod[g]{K}{L} = 1.\]
	Noting finally that in the $(r,u)$ coordinates $K = \partial_u$, we have that the metric $g$ must take the form \eqref{eq:metricformnull} as claimed. 
\end{proof}

\section{Applications to relativity}\label{sec:Causal}

In the remainder of this paper, we will assume that $(M,\tilde{g})$ is a space-time. More precisely, we will assume that $(Q,g)$ is such that $Q$ is homeomorphic to $\Real^2$ and $g$ has Lorentzian signature $(-+)$, and $(F,h)$ is connected with Riemannian signature $(++\dots)$.

\begin{rmk}
	Note that the universal cover of any Lorentzian two manifold $(Q,g)$ is topologically $\Real^2$, as $\Sphere^2$ is ruled out for lack of non-vanishing vector fields \cite[Cor. 3.38]{BeEhEa1996}. 
\end{rmk}

For warped product space-times, causality properties of $(M,\tilde{g})$ and the base $(Q,g)$ are generally the same. For example, 
as a consequence of our assumptions, we have that
\begin{prop}[{\cite[Prop. 3.41, Thm.\ 3.43, Prop.\ 3.64]{BeEhEa1996}}]
	\begin{enumerate}
		\item Both $(Q,g)$ and $(M,\tilde{g})$ are stably causal and time-orientable
		\item There exists a double-null foliation: more precisely, there exist two linearly independent functions $u,v:Q\to\Real$ such that $\innerprod[g]{\D*{u}}{\D*{u}} = \innerprod[g]{\D*{v}}{\D*{v}} = 0$, and $\innerprod[g]{\D*{u}}{\D*{v}} < 0$. Alternatively, there exists two linearly independent non-vanishing null vector fields $L, N$ on $Q$ such that $\innerprod[g]{L}{N} < 0$. 
	\end{enumerate}
\end{prop}
Relative to the double-null coordinates $u,v$, the metric $g$ takes the conformal form
\begin{equation}\label{eq:doublenullmetric}
	- \Omega (\D*{u} \otimes \D*{v} + \D*{v} \otimes \D*{u})
\end{equation}
and hence a direct computation gives 
\begin{equation}\label{eq:metricwave}
	\Ric{g} =  \partial^2_{uv} \ln \Omega (\D*{u} \otimes \D*{v} + \D*{v} \otimes \D*{u}).
\end{equation}
The Laplace-Beltrami operator is
\begin{equation}\label{eq:LBop}
	\triangle_g = - 2 \Omega^{-1} \partial^2_{uv}.
\end{equation}

Another general result is
\begin{thm}[{\cite[Thm.\ 3.68, Thm.\ 3.69]{BeEhEa1996}}]
	\begin{enumerate}
		\item The warped product $(M,\tilde{g})$ is globally hyperbolic if and only if $(Q,g)$ is globally hyperbolic and $(F,h)$ is complete.
		\item If $(Q,g)$ is globally hyperbolic with Cauchy hypersurface $\Sigma$, and $(F,h)$ is complete, then $\Sigma\times F$ is a Cauchy hypersurface for $(M,\tilde{g})$. 
	\end{enumerate}
\end{thm}
We shall in general assume that $(F,h)$ is complete.

At times we will also invoke energy conditions; we recall their definitions. Observe that the dominant energy condition implies the null energy condition by continuity. 
\begin{defn}
	\begin{itemize}
		\item $(M,\tilde{g})$ is said to satisfy the \emph{dominant energy condition} if for every pair of causal vectors $X_1, X_2$ satisfying $\innerprod[\tilde{g}]{X_1}{X_2} < 0$, the Einstein tensor $\Einst{\tilde{g}}(X_1,X_2) \geq 0$. 
		\item $(M,\tilde{g})$ is said to satisfy the \emph{null energy condition} if for every null vector $L$, the Einstein tensor satisfies $\Einst{\tilde{g}}(L,L) \geq 0$. 
	\end{itemize}
\end{defn}

A critical concept in the discussion of gravitational collapse is the null expansion. In the setting of warped product space-times, the computations simplify. 
\begin{defn}
	Given $q\in Q$, and let $L, N$ be the two independent future-directed null vectors at $T_q Q$. We say that $q$ is
	\begin{itemize}
		\item a \emph{regular point} if $L(r) N(r) < 0$ (equivalently $\innerprod[g]{\D*r}{\D*r} > 0$); 
		\item \emph{anti-trapped} if $L(r) > 0$ and $N(r) > 0$;
		\item \emph{trapped} if $L(r) < 0$ and $N(r) < 0$; 
		\item \emph{marginally $L$-trapped} if $L(r) = 0$ and $N(r) \leq 0$; \emph{marginally $N$-trapped} if $L(r) \leq 0$ and $N(r) = 0$. (Marginal anti-trapping are defined analogously.)
	\end{itemize}
\end{defn}
Notice that while the null \emph{directions} are determined uniquely by the metric $g$, there is no canonical way to choose the ``lengths'' of the null vectors. 
Correspondingly, the numerical values of $L(r)$ and $N(r)$ depend on the chosen frame; but their \emph{signs} are invariant under rescaling by $\Real_+$. Thus the above definition is independent of the specific vectors $L$ and $N$ used in the computations. 

Let $L$ be a null geodesic vector field, we compute
\begin{equation}\label{eq:raychaudhuri}
	\covD_L (L(r)) = \covD^2_{L,L} r - \covD_{\covD_L L} r = - \frac{r}{n} G_Q(L,L) \leq 0,
\end{equation}
where we used \eqref{eq:EinsteinQ}, provided that the null energy condition holds. This is the
\emph{Raychaudhuri equation} in the context of warped product space-times, and implies that the null-expansion is monotonically decreasing along null geodesics to the future. 

Keeping $L$ as a null geodesic vector field, define $N$ to be the unique null vector field satisfying $\innerprod[g]{L}{N} = -1$. 
The normalization forces $\covD_L N = 0$. Under this setting, we have
\[
	\covD_L (N(r)) = \covD^2_{L,N} r = \frac{r}{n} G_Q(L,N) - \frac{1}{2n r} \Scal{h} + \frac{n-1}{2r} \innerprod[g]{\D*{r}}{\D*{r}}.
\] 
Observing that $\innerprod[g]{\D*{r}}{\D*{r}} = - 2 L(r) N(r)$ by our choice of normalization, we can re-write the equality as
\begin{equation}\label{eq:propexp}
	\covD_L (N(r^n)) =  r^n G_Q(L,N) - \frac{r^{n-2}}{2} \Scal{h}.
\end{equation} 
An immediate consequence is a special case of the black-hole topology theorem \cite{GalSch2006, Gallow2008}.
\begin{thm}[Black-hole topology] \label{thm:topology}
	Let $(M,\tilde{g})$ be a warped product space-time satisfying the dominant energy condition. 
	Let $\gamma$ denote a null geodesic segment on $(Q,g)$, and $L$ a smooth, future-directed tangent vector field to $\gamma$. 
	Let $N$ denote a smooth, complementary future-directed tangent vector field along (and transverse to) $\gamma$.
	Suppose $p, q\in \gamma$ are such that
	\begin{itemize}
		\item $q$ is to the future of $p$;
		\item $N(r) \leq 0$ at $q$; $N(r) > 0$ at $p$.
	\end{itemize}
	Then $\Scal{h} > 0$. 
\end{thm}
\begin{proof}
	Since the sign conditions are invariant under rescaling of of the vector fields, we can assume without loss of generality that $L$ is geodesic and $\innerprod[g]{L}{N} = -1$ along $\gamma$. Furthermore, as $r > 0$ by assumption, we thus have that $N(r^n)$ is non-positive at $q$ and positive at $p$.
	
	By the mean value theorem, there exists some point $p'$ on $\gamma$ between $p$ and $q$ such that $\covD_L N(r^n) < 0$. Applying \eqref{eq:propexp} at this point we have that
	\[ 0 >   G_Q(L,N) - \frac{1}{2r^2} \Scal{h}.\]
	By the dominant energy condition the first term is non-negative, and hence $\Scal{h} > 0$.   
\end{proof}
We note that $\gamma\times F$ is a null hypersurface in $(M,\tilde{g})$. A version in terms of space-like hypersurfaces (as in the case of \cite{GalSch2006, Gallow2008}) follow as an immediate consequence. 
\begin{cor}\label{cor:topology}
	Let $(M,\tilde{g})$ be a warped product space-time satisfying the dominant energy condition. Let $\sigma:[0,1]\to Q$ be a space-like curve. Let $L, N$ be two linearly-independent, future-directed null vector fields along $\sigma$, normalized such that $\innerprod[g]{\dot{\sigma}}{N} > 0$ and $\innerprod[g]{\dot{\sigma}}{L} < 0$. Suppose that $N(r) = 0$ at $\sigma(0)$ and $N(r) > 0$ along $\sigma((0,1])$. Then $\Scal{h} > 0$. 
\end{cor} 
\begin{proof}
	Let $\gamma$ be the null geodesic generated by $L$ through $q\eqdef \sigma(0)$. For $p$ to the past of $q$ and sufficiently close to $q$, the transverse null geodesic to $\gamma$ through $p$ can be extended to intersect $\sigma$ at some point $p' \neq q$. As $N(r)|_{p'} > 0$, by \eqref{eq:raychaudhuri} $N(r)|_p > 0$ also, and thus we can apply Theorem \ref{thm:topology}. 
\end{proof}

As a result, we can also rule out \emph{spatially periodic} solutions which contain regular regions, when $(F,h)$ has negative curvature. (When $\Scal{h} > 0$, a classical example of such a space-time is the maximally extended Schwarzschild-de-Sitter solution.) 
In fact, such manifolds are necessarily trivial or cosmological.
\begin{thm}\label{thm:negcurperiod}
	Let $(M,\tilde{g})$ be a globally hyperbolic warped product space-time satisfying the dominant energy condition, with $\Scal{h} \leq 0$. Suppose there exists a Cauchy hypersurface $\sigma$ of $(Q,g)$ and a fixed-point-free isometry $\phi:Q\to Q$ that preserves time-orientation, maps $\sigma$ into itself, and commutes with the warping function $r$. Then either
	\begin{itemize}
		\item $\D*{r}$ is \emph{time-like} everywhere;
		\item $(M,\tilde{g})$ is a Cartesian product, with $(Q,g)$ flat and $(F,h)$ Ricci-flat.
	\end{itemize}
	In particular, $(Q,g)$ can contain no regular points. 
\end{thm}
\begin{proof}
	Let $L$ and $N$ denote our pair of linearly-independent, future-directed null vector fields on $Q$. We first claim that $L(r)$ and $N(r)$ has constant sign on $Q$. Let $q\in \sigma$. It splits $\sigma$ into to components $\sigma_+$ and $\sigma_-$, the two rays emanating from $q$ with tangents whose scalar products against $N$ are respectively positive and negative. By the contrapositive of Corollary \ref{cor:topology}, if $N(r)|_q > 0$ then $N(r) |_{\sigma_-} > 0$; if $N(r) |_q < 0$ then $N(r) |_{\sigma_+} < 0$. As $\phi$ has no fixed points, the orbit of $\sigma_\pm$ under $\phi$ is the entirety of $\sigma$. This shows that the sign of $N(r)$ must be constant along $\sigma$. By global hyperbolicity and the monotonicty formulae \eqref{eq:raychaudhuri} (which states $\covD_N N(r) \leq 0$) and \eqref{eq:propexp} (which states $\covD_L N(r^n) \geq 0$, using the non-positive curvature condition), the sign of $N(r)$ must be constant then along $Q$. 

	Now suppose that $\D*{r}$ is not time-like, then $L(r) N(r) \leq 0$. There are two possibilities. First, suppose every point in $Q$ is regular. But at regular points $\D*{r}$ is space-like, and in particular $r$ must be strictly monotonic along $\sigma$. This contradicts the existence of $\phi:\sigma\to\sigma$ with no fixed points that preserves $r$ in the hypotheses. Hence $(Q,g)$ can contain no regular point. 

	Similarly, suppose $N(r) = 0$ and $L(r) \neq 0$. Then $(\D*{r})^\sharp = \lambda N$ with $\lambda \neq 0$. This implies also that $r$ must be strictly monotonic along $\sigma$, which contradicts the existence of $\phi$. Hence we conclude that the only alternative to $\D*{r}$ being time-like is if $L(r) = N(r) = 0$ everywhere, and hence $r$ is constant. Returning to \eqref{eq:EinsteinQ}, we see that in this case, for the dominant energy condition to hold, we must have $\Scal{h} \geq 0$, combined with our assumptions this implies $\Scal{h} = 0$ and $G_Q = 0$. Since $G_Q = 0$, dominant energy condition forces $G_F = 0$, which by \eqref{eq:EinsteinF} requires $\Einst{h} \propto \Scal{g} h$. In dimension $n > 2$, this implies $(F,h)$ is Einstein. And hence for all dimensions $n \geq 2$ we have, by our conclusion that $\Scal{h} = 0$ that $(F,h)$ is Ricci-flat; \eqref{eq:EinsteinF} then implies $\Scal{g} = 0$. 
\end{proof}

The previous theorem classifies the spatially closed warped product space-times with $\Scal{h} \leq 0$. For the open case one would generally want to impose asymptotic conditions. 
We consider in particular the case of asymptotically conical space-times. 
\begin{defn}\label{defn:AC}
	Let $(M,\tilde{g})$ be a globally hyperbolic warped product space-time. We say that it is \emph{asymptotically conical} if $(F,h)$ is a closed Riemannian manifold, and there exists a Cauchy hypersurface $\sigma$ of $(Q,g)$ such that
	\begin{enumerate}
		\item $\sigma$ can be written as the disjoint union of $\sigma_-$, $\sigma_0$, and $\sigma_+$ where $\sigma_0$ is compact and $\sigma_\pm$ are diffeomorphic to $\Real$. 
		\item On $\sigma_{\pm}$ the function $r$ functions as a coordinate taking values in $(r_{\pm},\infty)$. 
		\item The induced metric on $\sigma_{\pm}$ converges to $\D*{r} \otimes \D*{r}$ as $r\to \infty$.
		\item The geodesic curvature of $\sigma_{\pm}$ converges to $0$ as $r \to \infty$. 
	\end{enumerate}
\end{defn}
\begin{rmk}\label{rmk:twoended}
	We assumed our space-time is two-ended. With the exception of $(F,h)$ being $\Sphere^n$ with the standard metric, in which case we are in the usual spherically-symmetric setting, a warped product space-time cannot have ``only one end''. (In the case of the standard sphere there is the possibility of closing up the points at which $r \searrow 0$.)
\end{rmk}
An immediate consequence of our definition is that
\begin{prop}\label{prop:negcuropen}
	There are no asymptotically conical warped product space-times with $\Scal{h} \leq 0$ and satisfying the dominant energy condition. 
\end{prop}
\begin{proof}
	Let $L, N$ be the two null vector fields along the Cauchy hypersurface $\sigma$. By definition we must have (up to exchanging the labels $L$ and $N$) on $\sigma_+$ a point $p_+$ such that $N(r) > 0$ and $L(r) < 0$, and similarly a point $p_-$ on $\sigma_-$ such that $N(r) < 0$ and $L(r) > 0$. But as was shown in the proof of Theorem \ref{thm:negcurperiod}, the non-positivity of $\Scal{h}$ would force, based on the existence of $p_+$, that $N(r) > 0$ on both $\sigma_0$ and $\sigma_-$, giving a contradiction. 
\end{proof}

For warped product manifolds with $\Scal{h} > 0$, the analogous statement is
\begin{prop}
	Let $(M,\tilde{g})$ be a globally hyperbolic warped product space-time satisfying the dominant energy condition. 
	Let $\sigma$ denote a Cauchy hypersurface of $(Q,g)$. 
	If either 
	\begin{itemize}
		\item $(M,\tilde{g})$ is spatially periodic as in Theorem \ref{thm:negcurperiod}; or
		\item $(M,\tilde{g})$ is asymptotically conical,
	\end{itemize} 
	then $\sigma$ cannot consist entirely of regular points. 
\end{prop}
We omit the proof as it is very similar to and simpler than the proofs of Theorem \ref{thm:negcurperiod} and Proposition \ref{prop:negcuropen}; but remark that the obvious counterexample is ruled out by Remark \ref{rmk:twoended}. 

To conclude this section, we remark that the asymptotic conical assumption implies that, by possibly enlarging $\sigma_0$, we can assume $\sigma_+$ consists entirely of regular points. Let $\gamma$ be the future-directed null geodesic ray emanating from $\partial\sigma_+$ with $\covD_{\dot{\gamma}} r |_{\gamma(0)} < 0$, then within the future development of $\gamma \cup \sigma_+$ we have that Dafermos' condition $\tilde{\mathbf{A}}'$ \cite[Sec.\ 6]{Daferm2005} holds, and hence the result concerning completeness of future null infinity applies equally. 

We add further that in a system where the initial data for all matter fields \emph{except} for electromagnetism have compact support, and the cosmological constant $\Lambda$ vanishes, Dafermos' condition $\mathbf{E}'$ on the non-emptiness of future null infinity follows from Theorem \ref{thm:birkhoff2}; see also the next section.

\section{Matter models}\label{sec:Matter}
In the following we will assume that our space-time $(M,\tilde{g})$ solves Einstein's equation 
\begin{equation}\label{eq:einstein}
	\Einst{\tilde{g}} = T - \Lambda \tilde{g},
\end{equation}
where $\Lambda$ is the cosmological constant, and $T$ is the stress-energy tensor for the matter content. This can be re-written in terms of the Ricci tensor as
\begin{equation}\label{eq:einstein2}
	\Ric{\tilde{g}} - \frac{2}{n} \Lambda \tilde{g} = T - \frac1n (\trace_{\tilde{g}} T) \tilde{g}.
\end{equation}
The warped product structure on $(M,\tilde{g})$ forces the Einstein tensor to be block-diagonal, and hence through \eqref{eq:einstein} we can also write $T = T_Q + T_F$, where $T_Q$ is tangent to the base $Q$ and $T_F$ is tangent to the fibers $F$. 
Einstein's equation then forces 
\begin{subequations}
	\begin{gather}
		\mathrm{div}_h T_F  = 0,  \label{eq:divTF} \\
		(\mathrm{div}_g T_Q)(X) + \frac{n}{r} T_Q( (\D*r)^\sharp,X) - (\trace_{r^2 h} T_F) \frac{X(r)}{r} = 0. \label{eq:divTQ}
	\end{gather}
\end{subequations}

\subsection{Electrovacuum space-times} 
Let us consider first the case that the space-time is either devoid of matter content, or that the only matter field is electromagnetic; lessons from the spherical symmetric setting show us that these two cases can often be bundled with similar features. 

Writing $H$ for the Faraday tensor in linear electrodynamics, the stress-energy tensor can be written in index notation as 
\begin{equation}
	T_{ab} = H_{ac} H_b{}^c - \frac14 \tilde{g}_{ab} H_{cd} H^{cd}.
\end{equation}
We make the following assumption:
\begin{equation}\label{eq:HisLift}
	H \text{ is the lift of a two form from } Q.
\end{equation}
In the case of spherical symmetry this assumption holds automatically (except when $n = 2$ where one can also have a factor proportional to $\mathrm{dvol}_h$). As $Q$ is two-dimensional, this means that $H = \alpha~\mathrm{dvol}_g$ for some scalar function $\alpha$. This in particular implies that $H_{ac} H_b^{c} = - \alpha^2 g_{ab}$ (see Proposition \ref{prop:hodge2}), and
therefore the stress-energy tensor take the form 
\begin{equation}
	T = - \frac12 \alpha^2 g + \frac12 \alpha^2 r^2 h.
\end{equation}
Thus we can apply Birkhoff's theorem (see Proposition \ref{prop:birkhoff1} and Theorem \ref{thm:birkhoff2}) with $\kappa = -1$ to obtain a complete description of such solutions. 
\begin{cor}\label{cor:BirkhoffEV}
	Let $(M,\tilde{g})$ be a warped product electro-vacuum space-time, which we assume satisfies \eqref{eq:HisLift}. Then there exists a constant $\alpha_0$ such that 
	\[ H = \frac{\alpha_0}{r^n} ~\mathrm{dvol}_g.\]
	and that 
	\[ \innerprod[g]{\D*r}{\D*r} = - \innerprod[g]{K}{K} = \frac{\Scal{h}}{n(n-1)} - \frac{2m}{r^{n-1}} - \frac{2\Lambda}{n(n+1)} r^2 + \frac{\alpha_0^2}{n(n-1) r^{2n-2}}. \] 
	Furthermore, exactly one of the following holds:
	\begin{enumerate}
		\item $H$, $\Lambda$, and $m$ vanish identically, $(Q,g)$ is locally isometric to $\Real^{1,1}$, $(F,h)$ is Ricci-flat, and $r$ is equal to one of the canonical optical functions on $\Real^{1,1}$. 
		\item The function $r$ is constant and the two-form $H$ is covariantly constant on $Q$. The manifold $(Q,g)$ is maximally symmetric and $(F,h)$ is Einstein. The constants $\alpha_0, r, m, \Lambda, \Scal{h}, \Scal{g}$ satisfies
			\begin{subequations}
			\begin{gather}
				\label{eq:constantr1} \left( \frac1n - \frac12\right) \Scal{h} = \frac{r^2}{2} \left( \Scal{g}  + \frac{\alpha_0^2}{r^{2n}} - 2\Lambda \right), \\
				\label{eq:constantr2} \frac{\Scal{h}}{2r^2} = \frac12 \frac{\alpha_0^2}{r^{2n}} + \Lambda,\\
				\innerprod[g]{\D*{r}}{\D*{r}} = 0.
			\end{gather}
			\end{subequations}
		\item The Kodama vector field $K$ is a non-trivial Killing field, the metric $g$ can be written in Eddington-Finkelstein form 
			\[ - \innerprod[g]{\D*r}{\D*r} \D{u^2} \pm (\D*{u} \otimes \D*r + \D*r \otimes \D*u).\]
			This class includes the special cases of the Reissner-Nordstr\"om family (and their de Sitter and anti de Sitter counterparts) with $(F,h)$ being the standard $\Sphere^n$. 
	\end{enumerate}
\end{cor}

\begin{rmk}\label{rmk:BirkhoffEV}
	Corollary \ref{cor:BirkhoffEV} is compatible with Theorem \ref{thm:negcurperiod}. Consider in particular the section option: For the dominant energy condition to hold necessarily $\frac12 \alpha_0^2 + \Lambda \geq 0$. But if $\Scal{h} \leq 0$, the only possibility is for $\Scal{h} = 0$ and $\Lambda = - \frac12 \alpha_0^2$ in \eqref{eq:constantr2}. Plugging this into \eqref{eq:constantr1} this implies $\Scal{g} \leq 0$. But the two dimensional anti-de-Sitter space is not globally hyperbolic, hence the only possibility remaining is that $\Scal{g} = 0$ and $\Lambda = \alpha_0 = 0$. Note that this also implies $m = 0$. 
\end{rmk}

\subsection{Scalar fields}
For scalar fields $\phi:M\to\Real$ satisfying the linear wave equation $\triangle_{\tilde{g}} \phi = 0$, the stress-energy tensor is given by 
\begin{equation}
	T = \D*\phi \otimes \D*\phi - \frac12 \tilde{g} \innerprod[\tilde{g}]{\D*\phi}{\D*\phi}.
\end{equation}
This gives 
\begin{equation} 
	T - \frac{1}{n}(\trace_{\tilde{g}} T) \tilde{g} = \D*\phi \otimes \D*\phi = \Ric{\tilde{g}} - \frac2n \Lambda \tilde{g} 
\end{equation}
by \eqref{eq:einstein2}. That the Ricci tensor is block diagonal implies that $\phi$ is either constant along $Q$, or constant along $F$. 

In the case $\phi$ is constant along $Q$, by \eqref{eq:RicciF}, we must have 
\begin{equation}\label{eq:scalfield1}
	r \triangle_g r + (n-1) \innerprod[g]{\D*r}{\D*r} = c 
\end{equation}
for some constant $c$, as the other terms are independent of $Q$. This implies then $(F,h,\phi)$ is a solution of the \emph{Riemannian} Einstein-scalar-field problem $\Ric{h} - c h = \D*\phi \otimes \D*\phi$. However, this requires $\D*\phi \cdot \triangle_h \phi = 0$, and in the case of $(F,h)$ being a closed Riemannian manifold, or in the case of $(F,h)$ being open and we impose suitable decay conditions on $\phi$, this implies the field $\phi$ is constant. From this we conclude that $T \equiv 0$ and we are reduced to the vacuum case. 

In the case $\phi$ is constant along $F$, again we must have \eqref{eq:scalfield1} holding, which implies $(F,h)$ is Einstein. And the equations of motion reduce similarly to the case of spherical symmetry. The fundamental unknowns are the metric $g$ (which after fixing a double-null coordinate system has one degree of freedom; see \eqref{eq:doublenullmetric}), the warping function $r$, and the scalar field $\phi$, which now satisfies \eqref{eq:divTQ}. It is convenient to augment the system with the Hawking energy. The equations are given by the following quasilinear system: 
\begin{subequations}
	\begin{gather}
		\mathrm{div}_g( r^n \D\phi) = 0 ,\\
		\triangle_g r = \frac{2(n-1)}{r^n} \hawking - \frac{2}{n} \Lambda r ,\\
		\Scal{g} = \frac{2n(n-1)}{r^{n+1}} \hawking + \frac{4-2n}{n} \Lambda + \innerprod[g]{\D*\phi}{\D*\phi},\\
		- \frac{n}{r} (\covD^2 r - \frac12 \triangle_g r g) = \D*{\phi} \times \D*{\phi} - \frac12 \innerprod[g]{\D*\phi}{\D*\phi},\\
		\D*\hawking = \frac{r^n}{n} \left[ \innerprod[g]{\D*{\phi}}{\D*{r}} \D\phi + \left( \Lambda - \frac12 \innerprod[g]{\D*\phi}{\D*\phi}\right) \D r \right].
	\end{gather}
\end{subequations}
In the null coordinates $u,v$, the equations can be expressed using \eqref{eq:doublenullmetric}, \eqref{eq:metricwave}, and \eqref{eq:LBop}; furthermore, following an observation of Christodoulou \cite{Christ1993}, the metric factor $\Omega$ can be eliminated (algebraically) by the formula
\[ \Omega = - \frac{2 r_{,u} r_{,v}}{ \tilde{S} - \frac{2 \hawking}{r^{n-1}}}, \qquad\text{ where } \tilde{S} \eqdef \frac{\Scal{h}}{n(n-1)}. \]
That the metric decouples is effectively due to the scalar field being a conformal field in two space-time dimensions.
We obtain the system (making use of the subscript notation for partial differentiation)
\begin{subequations}
	\begin{gather}
		2r\cdot \phi_{,uv} + n \left( r_{,u} \phi_{,v} + r_{,v} \phi_{,u}\right) = 0, \\
		r \cdot r_{,uv} = \frac{r_{,u} r_{,v}}{\tilde{S} - \frac{2 \hawking}{r^{n-1}}} \left[ (n-1) \frac{2\hawking}{r^{n-1}} - \frac{2}{n} \Lambda r^2 \right], \\
		2 r_{,u} \hawking_{,u} = \frac{r^n}{n} \left[ \left( \tilde{S} - \frac{2\hawking}{r^{n-1}} \right) \left(\phi_{,u}\right)^2 + 2 (r_{,u})^2 \Lambda\right], \\
		2 r_{,v} \hawking_{,v} = \frac{r^n}{n} \left[ \left( \tilde{S} - \frac{2\hawking}{r^{n-1}} \right) \left(\phi_{,v}\right)^2 + 2 (r_{,v})^2 \Lambda \right].
	\end{gather}
\end{subequations}
Much of the analysis of Einstein-scalar-field systems in spherical symmetry can be carried out analogously for this system. For example, the following is an immediate consequence of the propagation equations for the Hawking energy, 
\begin{lem}[Monotonicity of Hawking energy] \label{lem:HEmon}
	\begin{itemize}
		\item When $\Lambda \geq 0$, in the regular region (namely $\tilde{S} - \frac{2\hawking}{r^{n-1}} > 0$), we have that $r_{,u}$ and $\hawking_{,u}$ have the same sign, and $r_{,v}$ and $\hawking_{,v}$ have the same sign.
		\item When $\Lambda \leq 0$, in trapped regions $\hawking$ is non-decreasing in both $u$ and $v$, and in anti-trapped regions $\hawking$ is non-increasing. 
	\end{itemize}
\end{lem} 
For asymptotically conical solutions, this monotonicity also immediately implies a corresponding \emph{Penrose inequality}; this statement holds in a much general setting, see the discussion in \cite[Sec.\ 5]{Daferm2005}. 

Another particular consequence of this monotonicity is that 
\begin{prop}
	Let $(M,\tilde{g},\phi)$ be a warped product space-time solving the Einstein-scalar-field system with $\Lambda \leq 0$ and $\Scal{h} \geq 0$. Then the trapped region is a \emph{future set}, i.e.\ if $q$ is trapped and $p$ is to the future of $q$, the $p$ is also trapped. In particular, the boundary of the trapped region is \emph{achronal}.
\end{prop}
\begin{proof}
	As $q$ is trapped, we have that $\tilde{S} - \frac{2\hawking}{r^{n-1}} < 0$ there, and $r_{,u}$ and $r_{,v}$ are both negative. As $\Scal{h} \geq 0$ we must have $\hawking > 0$. 
	From Lemma \ref{lem:HEmon} we have that $\hawking_{,u}$ and $\hawking_{,v}$ are both non-negative. 
	This in particular implies 
	\[ - \partial_u \frac{\hawking}{r^{n-1}} < 0, \quad - \partial_v \frac{\hawking}{r^{n-1}} < 0, \]
	and shows that every point to the causal future of $q$ satisfies $\tilde{S} - \frac{2\hawking}{r^{n-1}} < 0$. 
	
	That the past boundary of a future set is achronal is a standard fact \cite[Ch.\ 14, \P27]{ONeil1983}.
\end{proof}
This result is specific to the scalar-field model.
For general space-times, the geometry of the boundary of the trapped region can be much more complicated; see \cite{Kommem2013}. 
We remark here that this monotonicity property does \emph{not} rely on the dominant energy condition, as when $\Lambda \leq 0$ the dominant energy condition is generally violated (the null energy condition, however, holds).

\subsection{Fluids}
For a fluid with density $\rho$, pressure $p$, and velocity field $\xi$, the stress-energy tensor is given by 
\begin{equation}
	T = (p + \rho) \xi\otimes \xi + p \tilde{g}.
\end{equation}
As $\xi$ is required to be a unit time-like vector field, we have that the condition $T$ is block diagonal requires $\xi$ to be tangent to $Q$. This implies that $T_F$ is proportional to $h$. By Lemma \ref{lem:rigid} necessarily $p$ is constant along $F$. 

Furthermore, as the trace-less part of of $T_Q$ is $(p+\rho) \xi\otimes \xi + \frac12 (p+\rho) g$, which as we recall is equal to the trace-less part of $\frac{n}{r} \covD^2 r$, we must have that $\rho$ and $\xi$ are also independent of $F$, which then, by consideration of the pure-trace part ($\frac12 (p-\rho) g$) and \eqref{eq:EinsteinQ} gives us that $(F,h)$ has constant scalar curvature. So for any dimension $n \geq 2$ we can conclude that $(F,h)$ is Einstein. 

Therefore the equations of motion reduce similarly as in the case of spherical symmetry. 
The conservation law \eqref{eq:divTQ} implies
\[ \mathrm{div}_g T_Q + \frac{n}{r} (p + \rho) (\covD_\xi r) \xi = 0,\]
which leads us to the following system of equations
\begin{subequations}
	\begin{gather}
		\label{eq:fluids1Energy} \covD_{r^n \xi}(\rho) + (p + \rho) \mathrm{div}_g (r^n \xi)= 0, \\
		\label{eq:fluids1Momentum} (p + \rho) \covD_\xi \xi + \covD p + (\covD_\xi p)\xi = 0.
	\end{gather}
\end{subequations}
Making the usual assumption that our fluid is homentropic, we can postulate that the equation of state is captured in a relation $p = p(\rho)$.
Consider the pair of functions $\Xi(\rho)$ and $H(\rho)$ given by $\Xi' / \Xi  = 1 / (p + \rho)$ and $H' / H =  p' / (p + \rho)$. They are related to the per-particle enthalpy, and we note that 
\[ \frac{\Xi'}{\Xi} + \frac{H'}{H} = \ln( \Xi H) ' = \frac{1 + p'}{p+\rho} =  \ln(p+\rho) '\]
so we will postulate $\Xi H = p + \rho$. 
The equations of motion can be re-written as
\begin{subequations}
	\begin{gather}
		\label{eq:fluids2Energy} \mathrm{div}_g \left( r^n \Xi \xi \right) = 0,\\
		\label{eq:fluids2Momentum} \covD_{\xi} (H \xi) + \covD H = 0. 
	\end{gather}
\end{subequations}
Noting that $\innerprod[g]{\covD_X(H\xi)}{\xi} = - \covD_X H$, \eqref{eq:fluids2Momentum} implies that 
\[ \intprod_{\xi} \D* (H\xi^\flat) = 0.\]
As $\xi$ is the lift of a vector field from $Q$, and $H$ by our discussion above is a function on $Q$, we must have that the two-form $\D* (H\xi^\flat)$ is proportional to the volume form on $Q$, and the fact that its interior product with $\xi$ vanishes implies that
\begin{equation}
	\D*(H \xi^\flat) = 0.
\end{equation}
This equation gives the unsurprising fact that homentropic fluids on a warped product space-time must be \emph{irrotational} (the two form in question being precisely the vorticity for relativistic fluids). 

Combining with \eqref{eq:fluids2Energy} we see that $H \xi$ solves a nonlinear Hodge system. On simply connected domains that $H\xi$ is curl-free allows us to lift to a \emph{fluid potential} which then solves a quasilinear wave equation; see also \cite[Ch.\ 1]{Christ2007a}. 
Here, for simplicity, we specialize to the ultra-relativistic equations of state $p = \gamma \rho$ where $\gamma \geq 0$ and $\sqrt{\gamma}$ is the \emph{speed of sound}. 
Two specific cases with cosmological applications are $\gamma = 0$ (the dust model) and $\gamma = \frac{1}{n+1}$ (the radiation model, where $\trace_{\tilde{g}} T = 0$). 

In this setting we can solve 
\[ H = \rho^{\frac{\gamma}{1 + \gamma}}, \quad \Xi = (1 + \gamma) \rho^{\frac{1}{1 + \gamma}}. \]
The equations of motion become the following system of equations on $Q$:
\begin{subequations}
	\begin{gather}
		\mathrm{div}_g \left( r^n \rho^{\frac{1}{1+\gamma}} \xi \right) = 0, \label{eq:fluids3Energy}\\
		\D*( \rho^{\frac{\gamma}{1+\gamma}} \xi^\flat)  = 0, \label{eq:fluids3Momentum}\\
		r \triangle_g r - \left( \frac1n \Scal{h} - (n-1) \innerprod[g]{\D*r}{\D*r} \right) + \frac{2\Lambda}{n} r^2 + \frac{1-\gamma}{n} \rho r^2 = 0, \label{eq:fluids3Warp} \\
		- \frac{n}{r} \covD^2 r + \frac{n}{2r} \triangle_g r g = (1 + \gamma) \rho  \left( \xi\otimes\xi + \frac12 g\right), \label{eq:fluids3Hawk} \\
		\Scal{g} = \frac{n}r \triangle_g r + \frac{4}{n} \Lambda  + \left(\frac{2-2\gamma}{n} - \gamma - 1\right) \rho \label{eq:fluids3Metric} .
	\end{gather}
\end{subequations}
We note in the case of the dust model, the second equation implies that $\xi$ is geodesic. 

\subsection{Cosmological solutions} 
Now consider the cosmological setting where we have Einstein's equation coupled to an ultra-relativistic fluid. 
We assume that we are in the \emph{expanding} setting, where every point is anti-trapped and $r$ serves as a time-function. 

We assume further the homogeneity condition
\begin{equation}\label{eq:assmhom}
	\text{There exists a Killing vector field on $(Q,g)$ that preserves $\rho$, $r$, and $\xi$.}
\end{equation}
Note that we do \emph{not} assume that the fluid velocity $\xi$ is orthogonal to the level sets of the function $r$; as we will see below the orthogonality can be derived as a consequence of condition \eqref{eq:assmhom}.  

On $(Q,g)$ we choose a coordinate system $(r,s)$, where $\innerprod[g]{\D*r}{\D*s} = 0$; this can be obtained by setting $s$ constant along the integral curves of $\covD r$. 
The Killing vector field in \eqref{eq:assmhom} can then be identified with $\partial_s$. The unknowns (all as functions of $r$) are 
\begin{align*}
	\text{Components of the metric}: &~ - \frac{1}{\sigma} \D{r^2} + \alpha \D{s^2}; \\
	\text{Components of } \xi^\flat:&~ \xi_r \D{r} + \xi_s \D{s}; \\
	\text{The density}: &~ \rho.
\end{align*}
Note that $\sigma = - \innerprod[g]{\D*r}{\D*r}$. 

First observe that
\[ \covD^2_{X, \partial_r} r = - \frac{1}{\sigma} \covD^2_{X, \covD r} r = - \frac{1}{2\sigma} \covD_X (\innerprod[g]{\covD r}{\covD r}) = \frac1{2\sigma} \covD_X \sigma. \]
So \eqref{eq:fluids3Hawk} implies
\[ 0 = - \frac{n}{r} \covD^2_{\partial_s,\partial_r} r = \frac{\rho}{\sqrt{a\sigma}} \xi_r \xi_s \]
and hence we must conclude that $\xi_s \equiv 0$ and $\xi_r = - \frac{1}{\sqrt{\sigma}}$.  Then \eqref{eq:fluids3Momentum} is satisfied automatically, and we note that $\xi$, being of constant norm and orthogonal to a Killing field $\partial_s$, must be geodesic. 

So we have reduced to the three unknowns $\alpha,\sigma, \rho$. The equation \eqref{eq:fluids3Energy} becomes
\[
	\frac{\partial}{\partial r} \left( r^n \rho^{\frac{1}{1+\gamma}} \sqrt{\alpha} \right) = 0,
\]
implying 
\[ 
	\rho = \frac{ \rho_0}{r^{n(1+\gamma)} \alpha^{\frac{1+\gamma}{2}}}
\]
for some constant $\rho_0 \geq 0$. Using \eqref{eq:fluids3Hawk} again we have
\[ \frac1{2\sigma} \partial_r \sigma = \covD^2_{\partial_r, \partial_r} r = - \frac1{2\sigma} \triangle_g r - \frac{1+\gamma}{2\sigma} \frac{r\rho}{n} \]
which we simplify to
\begin{equation}\label{eq:fluids4-2-1}
	\partial_r \sigma = - \triangle_g r - \frac{1+\gamma}{n} r\rho.
\end{equation}
As we can express
\[ \triangle_g r = - \sqrt{\frac{\sigma}{\alpha}} \frac{\partial}{\partial r} ( \sqrt{\alpha \sigma} ) = - \frac12 \partial_r \sigma - \frac12 \frac{\sigma}{\alpha} \partial_r \alpha,\]
So we can rewrite \eqref{eq:fluids4-2-1} as 
\[
		\frac12 \left( \partial_r \sigma - \frac{\sigma}{\alpha} \partial_r \alpha\right) = - \frac{(1+\gamma)}{n} r\rho.
\]
From \eqref{eq:fluids3Warp} we also get
\[
	- r\triangle_g r = \frac{r}{2}\left( \partial_r \sigma + \frac{\sigma}{\alpha} \partial_r \alpha\right)  = 
	    - \frac1n \Scal{h} - (n-1) \sigma + \frac{2\Lambda}{n} r^2 + \frac{1-\gamma}{n} r^2 \rho.
\]
Altogether, we have the following system of equations:
	\begin{gather*}
		\rho = \frac{ \rho_0}{r^{n(1+\gamma)} \alpha^{\frac{1+\gamma}{2}}}, \\
		\partial_r (r^{n-1} \sigma) = - \frac{1}{n} \Scal{h} r^{n-2}  + \frac{2\Lambda}{n} r^n - \frac{2\gamma}{n} \rho r^n ,\\
		\frac12 \partial_r \left( \frac{\sigma}{\alpha} \right) = - \frac{1 + \gamma}{n} \frac{r\rho}{\alpha}. 
	\end{gather*}
These can be further simplified to 
\begin{subequations}
	\begin{gather}
		\partial_r \left( r^{n-1} \sigma + \frac{\Scal{h}}{n(n-1)} r^{n-1} - \frac{2\Lambda}{n(n+1)} r^{n+1} \right) = - \frac{2\gamma \rho_0}{n r^{n\gamma} \alpha^{\frac{1 + \gamma}{2}}}, \label{eq:fluids5-1}\\
		\partial_r \left( \frac{\alpha}{\sigma}\right)^{\frac{1 + \gamma}2} = \frac{(1+\gamma)^2\rho_0}{n} r^{1 - n(1+\gamma)} \sigma^{- \frac{3+\gamma}{2}}. \label{eq:fluids5-2}
	\end{gather}
\end{subequations} 
We remark that the scalar curvature can be computed by combining \eqref{eq:fluids3Warp} and \eqref{eq:fluids3Metric}, with the formula
\begin{equation}\label{eq:Scalcurv}
	\Scal{g} = \frac{1}{r^2} \left( \Scal{h} + n(n+1)\sigma\right)  + \frac{4- 2n}{n} \Lambda + \left( \frac{2 - 2\gamma-2n}{n}\right) \rho. 
\end{equation}

Consider first the dust case where $\gamma = 0$. We can solve \eqref{eq:fluids5-1} to get
\begin{equation}\label{eq:dustSigma}
	\sigma = \frac{2m}{r^{n-1}} + \frac{2\Lambda}{n(n+1)} r^2 - \frac{\Scal{h}}{n(n-1)}
\end{equation}
for some constant $m$; compare to \eqref{eq:Knorm}. With the explicit form of $\sigma$, 
the right hand side of \eqref{eq:fluids5-2} can be integrated explicitly as it is a rational function of $r$. We shall not perform the explicit integration here, but just note the following asymptotic properties as $r\nearrow \infty$:
\begin{description}
	\item[$\Lambda > 0$:] $\sigma$ grows quadratically in $r$, this implies that $\alpha/\sigma$ converges to a fixed constant. By rescaling $s$ we can assume that this constant is $1$. This shows that the geometry converges toward the vacuum geometry \eqref{eq:metricformorth} of Theorem \ref{thm:birkhoff2}. This reflects the fact that for the dust fluid the density $\rho$ decays as the inverse of the spatial volume $ \sqrt{a} r^n$. 
	\item[$\Lambda = 0$:] when the cosmological constant vanishes, we have to consider the next most significant term, which is $\Scal{h}$. 
	\begin{description}
		\item[$\Scal{h} < 0$:] The spacetime continues to expand as $r\to \infty$. When dimension $n > 2$ again we have that $\alpha / \sigma$ converges to a fixed constant, and we again converges to the vacuum geometry of \eqref{eq:metricformorth}; the expansion only occurs in the fiber directions but not on $Q$. 
		
		 When dimension $n = 2$, however, $\alpha / \sigma$ grows logarithmically, and departs from the vacuum geometry. 
		 \item[$\Scal{h} > 0$:] In this case we must have $m > 0$ to admit a trapped/anti-trapped region. Therefore for some finite $r$ we have $\sigma = 0$, signaling the arrival at a cosmological horizon (similarly to inside the white-hole region of the extended Schwarzschild solution). 
		 \item[$\Scal{h} = 0$:] In this case, for the solution to be expanding, we must have the mass parameter $m > 0$. The equation \eqref{eq:fluids5-2} can be integrated explicitly to obtain that $\alpha$ grows, at leading order, like $r^2$. 
	\end{description}
\end{description}

In the case $\Lambda > 0$, we have that the right hand side of \eqref{eq:fluids5-2} is integrable near infinity. 
In this case it is also interesting to analyze the asymptotic behavior to the past. We consider specifically the case where $\Scal{h} \leq 0$ and $m \geq 0$, as in this case $\sigma$ is positive for all $r \in \Real_+$.  There are two basic cases:
\begin{description}
	\item[$m = 0$:] The right hand side of \eqref{eq:fluids5-2} is not integrable near the origin. 
	Thus there must exist some $r_0 > 0$ such that $\lim_{r\searrow r_0} \alpha = 0$.  By \eqref{eq:Scalcurv} this is a finite time curvature singularity (with $\abs{\Scal{g}} \approx \alpha^{-1/2}$) resulting from the collapse of only one direction.
	\item[$m > 0$:] The right hand side of \eqref{eq:fluids5-2} is integrable near the origin. In this case the past asymptotic behavior depends on the constant of integration (equivalently the limit at $r \nearrow \infty$), the three possibilities are
		\begin{enumerate}
			\item $\alpha$ vanishes at some finite $r_0$, the asymptotics are the same as described above. 
			\item $\alpha$ has finite positive limit as $r \searrow 0$. These solutions are past-time-like geodesically incomplete, with a scalar curvature singularity at $r \searrow 0$ with rate $\Scal{g} \approx r^{-(n+1)}$. As $\alpha$ grows like $r^{1-n}$, the singularity is similar to that of the Schwarzschild singularity, and not to the big bang singularities of homogeneous isotropic cosmological solutions.
			\item $\lim_{r \searrow 0} \alpha/\sigma = 0$. The equation \eqref{eq:fluids5-2} require that the derivative of $\sqrt{\alpha/\sigma}$ behaves as $r^{(n-1)/2}$ near the origin (since $\sigma$ diverges as $r^{1-n}$). This implies that $\alpha$ vanishes quadratically at the origin. In this case all spatial directions collapse at the same rate as $r \searrow 0$, resulting in a curvature blow-up. 
			
			This is the case that is most similar to the FLRW solutions with positive cosmological constants. Indeed, consider the case $(F,h) = \mathbb{T}^n$ with the standard (flat) metric, we observe that a particular pair of $(\alpha,\sigma)$ that solves \eqref{eq:fluids5-1} and \eqref{eq:fluids5-2} are given by 
			\[ \alpha = r^2, \qquad \sigma = \frac{2m}{r^{n-1}} + \frac{2\Lambda}{n(n+1)} r^2, \qquad 2n(n+1)m = \rho_0\]
			in which case the solution is precisely the FLRW solution with spatial geometry $\mathbb{T}^{n+1}$. 
		\end{enumerate}
\end{description}

For $\gamma > 0$, in the setting where $\Lambda > 0$, we can show the existence of future-global solutions with the same qualitative asymptotic behavior as the $\gamma = 0$ case. We quickly sketch the iterative argument. (Note that we make no assumptions on the size of $\gamma$.)
\begin{enumerate}
	\item Let $\sigma_0$ be the solution to \eqref{eq:fluids5-1} with $\gamma = 0$. 
	\item Let $\alpha_0$ solve \eqref{eq:fluids5-2} with $\sigma_0$ in the role of $\sigma$. Note that $\alpha_0$ has growth rate on the order of $r^2$.
	\item Let $\sigma_1$ solve \eqref{eq:fluids5-1} with $\alpha_0$ playing the role of $\alpha$. The rate of growth of $\alpha_0$ guarantees that the right hand side is integrable near infinity, and hence $\sigma_0 - \sigma_1$ decays at infinity. 
	\item Let $\alpha_1$ solve \eqref{eq:fluids5-2} with $\sigma_1$ in the role of $\sigma$. Note that $\alpha_1$ still has growth rate on the order of $r^2$. 
	\item Iterate the above process to construct $\sigma_i$ and $\alpha_i$ for $i > 1$. Restricting to the region $r \in (r_0, \infty)$, by choosing $r_0$ sufficiently large we see that the mapping from $(\sigma_i, \alpha_i) \mapsto (\sigma_{i+1}, \alpha_{i+1})$ is Lipschitz with small Lipschitz constant, and hence the sequence is Cauchy and converges. 
\end{enumerate}
We note that the same argument cannot be applied to the cases where $\Lambda = 0$, as the growth rates of $\alpha$ and/or $\sigma$ are unsuitable for controlling the errors. 

\appendix

\section{Properties of Hodge operator} \label{app:hodgeprop}
Here we summarize some standard properties of the Hodge operator on a two-dimensional pseudo-Riemannian manifold $(Q,g)$. 
The proofs (which we omit) are largely exercises in the linear algebra of $2\times 2$ matrices. 

\begin{defn}
	Let $\epsilon$ denote the volume two-form (after a choice of orientation) of $(Q,g)$. We define Hodge operator $*_g$ by the following formulas.
	\begin{enumerate}
		\item Let $\eta$ be a one-form; then the Hodge operator acts as 
			\[ *_g \eta =  \intprod_{\eta^\sharp} \epsilon. \]
		\item Let $X$ be a vector field; then the Hodge operator acts as
			\[ *_g X = (\intprod_X \epsilon)^\sharp.\]
		\item Let $B$ be a symmetric covariant two-tensor, we define the left and right Hodge duals to be respectively
			\[ (*_g B)(X,Y) = - B(*_g X,Y), \quad (B*_g)(X,Y) = - B(X, *_g Y).\]
	\end{enumerate}
\end{defn}

\begin{prop}[Antisymmetry]
	The Hodge operator is antisymmetric with respect to the metric.
	\begin{enumerate}
		\item Let $\eta, \tau$ be one-forms; then $\innerprod[g]{*_g \eta}{\tau} + \innerprod[g]{\eta}{*_g\tau} = 0$. 
		\item Let $X,Y$ be vector fields; then $\innerprod[g]{*_g X}{Y} + \innerprod[g]{X}{*_g Y} = 0$.
		\item Let $B, D$ be symmetric covariant two-tensors; then 
			\[ \innerprod[g]{*_g B}{D} = \innerprod[g]{B*_g}{D} = - \innerprod[g]{B}{*_g D} = - \innerprod[g]{B}{D*_g}.\]
	\end{enumerate}
\end{prop}

\begin{defn}
	Define $\sign(g) = -1$ if $g$ has signature $(++)$ or $(--)$; and set $\sign(g) = +1$ if $g$ has signature $(+-)$ or $(-+)$. 
\end{defn}

\begin{prop}[Formulas involving taking the operator twice] \label{prop:hodge2}
	\begin{enumerate}
		\item $*_g *_g = \sign(g)$. 
		\item If $B$ is a symmetric covariant two-tensor, then 
			\[ *_g B *_g = \sign(g) \left[ B - (\trace_g B) g \right].\]
	\end{enumerate}
\end{prop}

\printbibliography
\end{document}